\documentclass{jmda}
\usepackage[spanish,english]{babel}

\usepackage{amsmath}
\usepackage{amsfonts}
\usepackage{latexsym}
\usepackage{epsfig}
\usepackage{graphicx}
\usepackage{url}

\begin{document}
\selectlanguage{english}

\title*{%
${\mathbb{Z}}_2{\mathbb{Z}}_4$-additive cyclic codes,\\ generator polynomials and dual codes
 \thanks{%
This work has been partially supported by the Spanish MINECO grant TIN2013-40524-P and by the Catalan AGAUR grant 2014SGR-691.
 }
}

\author{
 Joaquim Borges Ayats%
\and
 Cristina Fern\'andez-C\'ordoba %
\and
Roger Ten-Valls%
}

\institute{
 Department of Information and Communication Engineering, Universitat Aut\`{o}noma de Barcelona, 08193-Bellaterra, Spain.
  \url{{joaquim.borges, cristina.fernandez, roger.ten}@uab.cat}%
}
\maketitle


\begin{abstract}
A ${\mathbb{Z}}_2{\mathbb{Z}}_4$-additive code ${\cal C}\subseteq{\mathbb{Z}}_2^\alpha\times{\mathbb{Z}}_4^\beta$ is called cyclic if the set of coordinates can be partitioned into two subsets, the set of ${\mathbb{Z}}_2$ and the set of ${\mathbb{Z}}_4$ coordinates, such that any cyclic shift of the coordinates of both subsets leaves the code invariant. These codes can be identified as submodules of the $\mathbb{Z}_4[x]$-module $\mathbb{Z}_2[x]/(x^\alpha-1)\times\mathbb{Z}_4[x]/(x^\beta-1)$. The parameters of a ${\mathbb{Z}}_2{\mathbb{Z}}_4$-additive cyclic code are stated in terms of the degrees of the generator polynomials of the code. The generator polynomials of the dual code of a ${\mathbb{Z}}_2{\mathbb{Z}}_4$-additive cyclic code are determined in terms of the generator polynomials of the code ${\cal C}$.
\end{abstract}


\section{Introduction}

{D}{enote} by  ${\mathbb{Z}}_2$ and ${\mathbb{Z}}_4$ the rings of integers modulo 2 and modulo 4,
respectively.  We denote the space of $n$-tuples over these rings as ${\mathbb{Z}}_2^n$   and ${\mathbb{Z}}_4^n$.
A binary code is any non-empty subset $C$ of ${\mathbb{Z}}_2^n$. If that subset is a vector space then we say that it is a linear code. A code over $\mathbb{Z}_4$ is a non-empty subset ${\cal C}$ of ${\mathbb{Z}}_4^n$ and a submodule of ${\mathbb{Z}}_4^n$ is called a linear code over $\mathbb{Z}_4$.

In   Delsarte's 1973 paper (see
\cite{del}),  he defined additive codes as  subgroups of the underlying abelian
group in a translation association scheme. For the binary Hamming scheme, 
namely, when the underlying abelian group is
of order $2^{n}$, the only structures for the abelian group are
those of the form ${\mathbb{Z}}_2^\alpha\times {\mathbb{Z}}_4^\beta$, with
$\alpha+2\beta=n$. 
This means that  the subgroups ${\cal C}$ of
${\mathbb{Z}}_2^\alpha\times {\mathbb{Z}}_4^\beta$ are the only additive codes in a
binary Hamming scheme. In \cite{AddDual}, $\mathbb{Z}_2\mathbb{Z}_4$-additive codes were studied.

For vectors ${\bf  u} \in {\mathbb{Z}}_2^\alpha\times{\mathbb{Z}}_4^\beta$ we write
 ${\bf  u}=(u\mid  u')$ where
$ u=(u_0,\dots,u_{\alpha-1})\in{\mathbb{Z}}_2^\alpha$ and
$ u'=(u'_0,\dots,u'_{\beta-1})\in{\mathbb{Z}}_4^\beta$.

Let ${\cal C}$ be a ${\mathbb{Z}}_2{\mathbb{Z}}_4$-additive code. Since ${\cal 
C}$ is a subgroup of $\mathbb{Z}_2^\alpha\times\mathbb{Z}_4^\beta$, it is also 
isomorphic to a commutative structure like 
$\mathbb{Z}_2^\gamma\times\mathbb{Z}_4^\delta$.
Therefore, ${\cal C}$ is of type $2^\gamma 4^\delta$ as a group, it has $|{\cal 
C}| = 2^{\gamma +2\delta}$ codewords and the number of order two codewords in 
${\cal C}$ is $2^{\gamma +\delta}$.

Let $X$ (respectively $Y$) be the set of $\mathbb{Z}_2$ (respectively 
$\mathbb{Z}_4$) coordinate positions, so $|X| =\alpha$ and $|Y| = \beta$. Unless 
otherwise stated, the set $X$ corresponds to the first $\alpha$ coordinates and
$Y$ corresponds to the last $\beta$ coordinates. Call ${\cal C}_X$ 
(respectively ${\cal C}_Y )$ the punctured code of ${\cal C}$ by deleting the 
coordinates outside $X$ (respectively $Y$). Let ${\cal C}_b$ be the subcode of 
${\cal C}$ which contains all order two codewords and let $\kappa$ be the 
dimension of $({\cal C}_b)_X$, which is a binary linear code. For the case 
$\alpha = 0$, we will write $\kappa = 0$.

Considering all these parameters, we will say that ${\cal C}$ is of type 
$(\alpha, \beta; \gamma , \delta; \kappa)$. Notice that ${\cal C}_Y$ is a linear code over $\mathbb{Z}_4$ of type $(0, \beta; \gamma_Y , \delta; 0)$, where
$0 \leq \gamma_Y \leq \gamma$, and ${\cal C}_X$ is a binary linear code of type 
$(\alpha, 0; \gamma_X , 0; \gamma_X )$, where $\kappa \leq \gamma_X \leq\kappa + 
\delta$. A ${\mathbb{Z}_2 {\mathbb{Z}_4}}$-additive code ${\cal C}$ is said to 
be 
separable if ${\cal C} = {\cal C}_X \times {\cal C}_Y$. 

Let $\kappa_1$ and $\delta_2$ be the dimensions of the subcodes $\{(u \mid 0\dots 0)\in {\cal C}\}$ and $\{( 0\dots 0\mid u')\in {\cal C} : \mbox{ the order of }u'\mbox{ is } 4\}$, respectively. Define $\kappa_2=\kappa -\kappa_1$ and $\delta_1=\delta -\delta_2$. By definition, it is clear that a ${\mathbb{Z}_2{\mathbb{Z}_4}}$-additive code is separable if and only if $\kappa_2$ and $\delta_1$ are zero; that is, $\kappa=\kappa_1$ and $\delta=\delta_2$.

We define a Gray Map as $\phi: {\mathbb{Z}}_2^\alpha \times{\mathbb{Z}}_4^\beta\rightarrow\mathbb{Z}_2^{\alpha+2\beta}$ such that
$\phi(\textbf{u})=\phi( u \mid  u') = ( u, \phi_4( u'))$, where $\phi_4$ is the usual quaternary Gray map defined by 
$\phi_4(0) = (0,0), 
\phi_4(1) = (0,1), 
\phi_4(2) = (1,1), 
\phi_4(3) = (1,0).$

The \textit{standard inner product}, defined in \cite{AddDual}, can be written as 
$$\textbf{u}\cdot \textbf{v} = 2\left(\sum_{i=0}^{\alpha-1}u_iv_i\right)+\sum_{j=0}^{\beta-1}u'_jv'_j\in \mathbb{Z}_4,$$
where the computations are made taking the zeros and ones in the $\alpha$ binary coordinates as zeros and ones in $\mathbb{Z}_4$, respectively. The \textit{dual code} of ${\cal C}$, is defined in the standard way by
$${\cal C}^\perp=\{\textbf{v} \in \mathbb{Z}_2^\alpha\times\mathbb{Z}_4^\beta \mid \textbf{u}\cdot\textbf{v}=0, \mbox{ for all }\textbf{u}\in{\cal C}\}.$$

If ${\cal C}$ is separable then ${\cal C}^\perp=({\cal C}_X)^\perp\times ({\cal 
C}_Y)^\perp$. From \cite{AddDual}, and the previous definition of $\kappa_1$ and $\delta_1$ 
we obtain the number of codewords of ${\cal C}$, ${\cal C}_X$, ${\cal C}_Y$ and 
their duals.

\begin{proposition}\label{Dimension_subcodes}
Let ${\cal C}$ be a ${\mathbb{Z}_2 {\mathbb{Z}_4}}$-additive code of type 
$(\alpha, \beta; \gamma , \delta; \kappa)$. Let $\kappa_1$ and $\delta_1$ be defined as before. Then,
\begin{align*}
|{\cal C}|= 2^\gamma 4^\delta, & \quad |{\cal C}^\perp|= 
 2^{\alpha+\gamma-2\kappa}4^{\beta-\gamma-\delta+\kappa},\\
|{\cal C}_X|=2^{\kappa+\delta_1}, & \quad |({\cal 
C}_X)^\perp|=2^{\alpha-\kappa-\delta_1},\\
 |{\cal C}_Y|=2^{\gamma-\kappa_1}4^\delta , & \quad |({\cal 
C}_Y)^\perp|=2^{\gamma-\kappa_1}4^{\beta-\gamma-\delta+\kappa_1}.  
\end{align*}
\end{proposition}

Let ${\cal C}$ be a ${\mathbb{Z}_2 {\mathbb{Z}_4}}$-additive code of type 
$(\alpha, \beta; \gamma , \delta; \kappa)$. Then, ${\cal C}$ is permutation 
equivalent to a ${\mathbb{Z}_2 {\mathbb{Z}_4}}$-additive code with generator 
matrix of the form
{
\begin{equation*}
{\cal G}_{\cal C} =
 \left(\begin{array}{cccc|ccccc}
  I_{\kappa_1} & T & T'_{b_1} & T_{b_1} & 0 & 0 & 0 & 0 & 0 \\
  0 & I_{\kappa_2} & T'_{b_2} & T_{b_2} & 2T_2 & 2T_{\kappa_2} & 0 & 0 & 0 \\
  0 & 0 & 0 & 0 & 2T_1 & 2T'_1 & 2I_{\gamma-\kappa} & 0 & 0 \\
  \hline
  0& 0& S_{\delta_1} & S_{b} & S_{11} & S_{12} & R_1 & I_{\delta_1} & 0 \\
  0& 0& 0 & 0 & S_{21} & S_{22}& R_2 & R_{\delta_1} & I_{\delta_2}  
 \end{array}\right)
\end{equation*}
}
where $I_{r}$ is the identity matrix of size $r\times r$; the matrices $T_{b_i}, T'_{b_i}, S_{\delta_1}, S_b$
are over ${\mathbb{Z}}_2$; the matrices $T_1, T_2, T_{\kappa_2}, T'_1, R_i$ are over ${\mathbb{Z}}_4$ with all entries in $\{0,1\}\subset{\mathbb{Z}}_4$; and $S_{ij}$ are matrices over ${\mathbb{Z}}_4$. The matrices $S_{\delta_1}$ and $T_{\kappa_2}$ are square matrices of full rank $\delta_1$ and $\kappa_2$ respectively, $\kappa=\kappa_1+\kappa_2$ and $\delta=\delta_1+\delta_2$.

This new generator matrix can be obtained by applying convenient column permutations 
and linear combinations of rows to the generator matrix giving in 
\cite{AddDual}. This new form is going to help us to relate the parameters of 
the code and the degrees of the generator polynomials of a 
$\mathbb{Z}_2\mathbb{Z}_4$-additive cyclic code.

\section{${\mathbb{Z}_2 {\mathbb{Z}_4}}$-additive cyclic codes}

\subsection{Parameters and generators}

Let ${\bf  u} = ( u \mid  u') \in {\mathbb{Z}}_2^\alpha \times {\mathbb{Z}}_4^\beta$ and $i$ be an integer. Then we denote by
\begin{align*}
\textbf{u}^{(i)}&= (u^{(i)}\mid u'^{(i)})\\
&=(u_{0+i},u_{1+i},\dots,u_{\alpha-1+i}\mid u'_{0+i},u'_{1+i},\dots,u'_{\beta-1+i})
\end{align*}
the cyclic $i$th shift of $\textbf{u}$, where the subscripts are read modulo $\alpha$ and $\beta$, respectively. 

We say that a ${\mathbb{Z}_2 {\mathbb{Z}_4}}$-additive code ${\cal C}\subseteq\mathbb{Z}_2^\alpha\times\mathbb{Z}_4^\beta$ is \textit{cyclic} if for any codeword $\textbf{u}\in {\cal C}$ we have $\textbf{u}^{(1)}\in{\cal C}$. 

Let $R_{\alpha,\beta}=\mathbb{Z}_2[x]/(x^\alpha-1)\times\mathbb{Z}_4[x]/(x^\beta-1)$, for 
$\beta\geq 0$ odd and define the operation $\star : \mathbb{Z}_4[x]\times R_{\alpha,\beta}\rightarrow R_{\alpha,\beta}$ as $\lambda(x)\star (p(x) \mid q(x)) = (\lambda(x)p(x) \bmod (2) \mid \lambda(x)q(x))$. From \cite{Abu}, we know that ${\mathbb{Z}_2 
{\mathbb{Z}_4}}$-additive cyclic code are identified as $\mathbb{Z}_4[x]$-submodules of 
$R_{\alpha,\beta}$. Moreover, if ${\cal C}$ is a ${\mathbb{Z}_2 
{\mathbb{Z}_4}}$-additive cyclic code of type $(\alpha, \beta; \gamma , \delta; 
\kappa)$, then it is of the form 
\begin{equation} \label{form}
{\cal C}=\langle (b(x)\mid{ 0}), (\ell(x) \mid f(x)h(x) +2f(x)) \rangle  
\end{equation}
where $f(x)h(x)g(x) = x^\beta -1$ in ${\mathbb{Z}}_4[x]$, $b(x), \ell(x)\in\mathbb{Z}_2[x]/(x^\alpha-1)$ with $b(x)|(x^\alpha-1)$, $deg(\ell(x)) < deg(b(x))$ and $b(x)$ divides $\frac{x^\beta -1}{f(x)} \ell(x) \pmod{2}.$

Note that if ${\cal C}$ is a ${\mathbb{Z}_2 {\mathbb{Z}_4}}$-additive cyclic code with ${\cal C}=\langle (b(x)\mid { 0}), (\ell(x) \mid  f(x)h(x) +2f(x)) \rangle $, then the canonical projections ${\cal C}_X$ and ${\cal C}_Y$ are a cyclic code over $\mathbb{Z}_2$ and a cyclic code over $\mathbb{Z}_4$ generated by $gcd(b(x),\ell(x))$ and $(f(x)h(x)+2f(x))$, respectively (see \cite{macwilliams}, \cite{Wan}).

Since $b(x)$ divides $\frac{x^\beta -1}{f(x)} \ell(x) \pmod{2},$ we have the following result.

\begin{corollary}\label{bdiviXSgcd}
Let ${\cal C}$ be a ${\mathbb{Z}_2 {\mathbb{Z}_4}}$-additive cyclic code of type $(\alpha, \beta; \gamma , \delta; \kappa)$ with ${\cal C} = \langle (b(x)\mid { 0}), (\ell(x) \mid  f(x)h(x) +2f(x)) \rangle$. Then, $b(x)$ divides $\frac{x^\beta -1}{f(x)} \gcd(b(x),\ell(x)) \pmod{2}$ and $b(x)$ divides $h(x) \gcd(b(x),\ell(x)g(x)) \pmod{2}.$
\end{corollary}

Note that if a ${\mathbb{Z}_2 {\mathbb{Z}_4}}$-additive cyclic code is separable, then $\ell(x)=0$.

In the following, a polynomial $f(x)\in \mathbb{Z}_2[x]$ or $\mathbb{Z}_4[x]$ will be denoted simply by $f$ and the parameter $\beta$ will be an odd integer.

\begin{lemma}\label{generators_Cb}
Let ${\cal C}=\langle (b\mid { 0}), (\ell \mid  fh +2f) \rangle$ be a ${\mathbb{Z}_2 {\mathbb{Z}_4}}$-additive cyclic code. Then,
$${\cal C}_b=\langle (b\mid { 0}), (\ell g \mid  2fg), ({ 0}\mid  2fh) \rangle.$$
\end{lemma}
\begin{proof}
${\cal C}_b$ is the subcode of ${\cal C}$ which contains all codewords of order $2$. Since ${\cal C}= \langle (b\mid { 0}), (\ell \mid  fh +2f) \rangle$, then all codewords of order $2$ are generated by $\langle (b\mid { 0}), (\ell g \mid  2fg), ({ 0}\mid  2fh) \rangle.$
\end{proof}

The following results shows the close relation of the parameters of the type 
of a $\mathbb{Z}_2\mathbb{Z}_4$-additive cyclic code and the degrees of the 
generator polynomials of the code.

First, the next theorem gives the spanning sets in terms of the generator 
polynomials. 

\begin{theorem}{\cite[Theorem 13]{Abu}}\label{SpanningSetTh}
Let ${\cal C}=\langle (b\mid{ 0}), (\ell \mid fh +2f) \rangle$ be a 
\mbox{${\mathbb{Z}_2 {\mathbb{Z}_4}}$-additive} cyclic code of type $(\alpha, 
\beta; \gamma , \delta; \kappa)$, where $fhg = x^\beta -1$. Let

$$S_1=\bigcup^{\alpha-\deg(b)-1}_{i=0} \{x^i\star(b\mid 0)\},\quad
S_2=\bigcup^{\deg(g)-1}_{i=0}\{x^i\star(\ell \mid fh+2f)\} $$
and
$$S_3= \bigcup^{\deg(h)-1}_{i=0} \{x^i\star(\ell g\mid 2fg)\}.$$
Then, $S_1\cup S_2\cup S_3$ forms a minimal spanning set for ${\cal C}$ as a 
$\mathbb{Z}_4$-module. Moreover, $\cal C$ has 
$2^{\alpha-\deg(b)}4^{\deg{(g)}}2^{\deg{(h)}}$ codewords.
\end{theorem}

Note that $S_2$ generates all order $4$ codewords and the subcode of codewords 
of order 2,  ${\cal C}_b$, is generated by $\{S_1, 2S_2, S_3\}$. Hence, in
the following theorem, by using these spanning sets, we can obtain the 
parameters $(\alpha, \beta; \gamma , \delta; \kappa)$ of the code.

\begin{theorem}\label{TypeDependingDeg}
Let ${\cal C}= \langle (b\mid { 0}), (\ell \mid  fh +2f) \rangle $ be a ${\mathbb{Z}_2 {\mathbb{Z}_4}}$-additive cyclic code of type $(\alpha, \beta; \gamma , \delta; \kappa)$, where $fhg = x^\beta -1.$
Then 
\begin{align*}
\gamma &= \alpha -\deg(b)+\deg(h),\\
\delta &= \deg(g),\\
\kappa &= \alpha -\deg(\gcd(\ell g, b)).
\end{align*}
\end{theorem}
\begin{proof}
The parameters $\gamma$ and $\delta$ are known from Theorem~\ref{SpanningSetTh} 
and the parameter $\kappa$ is the dimension of $({\cal C}_b)_X$. By 
Lemma~\ref{generators_Cb}, the space $({\cal C}_b)_X$ is generated by the 
polynomials $b$ and $\ell g$.  Since the ring is a polynomial ring and thus a 
principal ideal ring, it is generated by the greatest common divisor of the 
two polynomials.  Then, $\kappa = \alpha -deg(gcd(\ell g,b)).$
\end{proof} 

In this case we have that $|{\cal C}| = 2^{\alpha - deg(b)} 4^{deg(g) } 2^{deg(h)}.$

\begin{proposition}\label{k1k2d1d2}
Let ${\cal C}=\langle (b\mid{ 0}), (\ell \mid fh +2f) \rangle$ be a \mbox{${\mathbb{Z}_2 {\mathbb{Z}_4}}$-additive} cyclic code of type $(\alpha, \beta; \gamma , \delta=\delta_1+\delta_2; \kappa=\kappa_1+\kappa_2)$, where $fhg = x^\beta -1$. Then,
$$\kappa_1= \alpha -\deg(b),\quad \kappa_2= \deg(b)-\deg(\gcd(b, \ell g)),$$
$$\delta_1= \deg(\gcd(b,\ell g)) - \deg(\gcd(b,\ell)) \mbox{ and } \delta_2=\deg(g)-\delta_1.$$
\end{proposition}
\begin{proof}
The result follows from Proposition \ref{Dimension_subcodes} and knowing the generator polynomials of ${\cal C}_X$ and $({\cal C}_b)_X$. They are $\gcd(b,\ell)$ and $\gcd(b, \ell g)$, respectively.
\end{proof}

\subsection{Dual ${\mathbb{Z}_2 {\mathbb{Z}_4}}$-Additive Cyclic Codes}
In \cite{Abu}, it is proven that the dual code of a $\mathbb{Z}_2\mathbb{Z}_4$-additive cyclic code is also a $\mathbb{Z}_2\mathbb{Z}_4$-additive cyclic code. So, we will denote
$${\cal C}^\perp = \langle (\bar{b}\mid{ 0}), (\bar{\ell} \mid \bar{f}\bar{h} +2\bar{f}) \rangle,$$
where $\bar{f}\bar{h}\bar{g} = x^\beta -1$ in ${\mathbb{Z}}_4[x]$, $\bar{b}, \bar{\ell}\in\mathbb{Z}_2[x]/(x^\alpha-1)$ with $\bar{b}|(x^\alpha-1)$, $deg(\bar{\ell}) < deg(\bar{b})$ and $\bar{b}$ divides $\frac{x^\beta -1}{\bar{f}} \bar{\ell} \pmod{2}.$

The \textit{reciprocal polynomial} of a polynomial $p(x)$ is $x^{\deg(p(x))}p(x^{-1})$ and is denoted by $p^*(x)$. As in the theory of cyclic codes over $\mathbb{Z}_2$ and $\mathbb{Z}_4$ (see \cite{macwilliams}, \cite{PQ}), reciprocal polynomials have an important role on duality.

We denote the polynomial $\sum^{m-1}_{i=0} x^i$ by $\theta_m(x)$. Using this notation we have the following proposition.

\begin{proposition}\label{xnm=xntheta}
Let $n,m\in \mathbb{N}$. Then, $$x^{nm}-1= (x^n-1)\theta_m(x^n).$$
\end{proposition}
\begin{proof}
It is well know that $y^m-1=(y-1)\theta_m(y)$, replacing $y$ by $x^n$ the result follows.
\end{proof}
From now on, $\mathfrak{m}$ denotes the least common multiple \mbox{of $\alpha$ and $\beta$.}
\begin{definition}
Let $\textbf{u}(x)=(u(x)\mid u'(x))$ and $\textbf{v}(x)=(v(x)\mid v'(x))$ be elements in $R_{\alpha,\beta}$. We define the map
$$\circ:R_{\alpha,\beta}\times R_{\alpha,\beta}\longrightarrow \mathbb{Z}_4[x]/(x^{\mathfrak{m}}-1),$$
such that
\begin{align*}
\circ(\textbf{u}&(x),\textbf{v}(x))= 2u(x)\theta_{\frac{\mathfrak{m}}{\alpha}}(x^\alpha)x^{\mathfrak{m}-1-\deg(v(x))}v^*(x) +\\
&+  u'(x)\theta_\frac{\mathfrak{m}}{\beta}(x^\beta)x^{\mathfrak{m}-1-\deg(v'(x))}{v'}^*(x) \!\!\mod(x^{\mathfrak{m}}-1),
\end{align*}
where the computations are made taking the binary zeros and ones in $u(x)$ and $v(x)$ as quaternary zeros and ones, respectively.
\end{definition}
The map $\circ$ is linear in each of its arguments; i.e., if we fix the first entry of the map invariant, while letting the second entry vary, then the result is a linear map. Similarly, when fixing the second entry invariant. Then, the map $\circ$ is a bilinear map between $\mathbb{Z}_4[x]$-modules.

From now on, we denote $\circ(\textbf{u}(x), \textbf{v}(x))$ by $\textbf{u}(x)\circ\textbf{v}(x)$. Note that $\textbf{u}(x)\circ\textbf{v}(x)$ belongs to $\mathbb{Z}_4[x]/(x^{\mathfrak{m}}-1)$.

\begin{proposition}
Let $\textbf{u}$ and $\textbf{v}$ be vectors in $\mathbb{Z}_2^{\alpha}\times\mathbb{Z}_4^{\beta}$ with associated polynomials $\textbf{u}(x)=(u(x)\mid u'(x))$ and $\textbf{v}(x)=(v(x)\mid v'(x))$. Then, $\textbf{u}$ is orthogonal to $\textbf{v}$ and all its shifts if and only if $$\textbf{u}(x)\circ\textbf{v}(x)= 0.$$ 
\end{proposition}
\begin{proof}
The $i$th shift of $\textbf{v}$ is $\textbf{v}^{(i)}=(v_{0+i}v_{1+i}\ldots v_{\alpha-1+i}\mid v'_{0+i}\ldots v'_{\beta-1+i})$. Then, 

$$\textbf{u}\cdot\textbf{v}^{(i)}=0\mbox{ if and only if }2\sum^{\alpha-1}_{j=0} u_jv_{j+i} +\sum^{\beta-1}_{k=0} u'_kv'_{k+i}=0.$$
Let $S_i=2\sum^{\alpha-1}_{j=0} u_jv_{j+i} +\sum^{\beta-1}_{k=0} u'_kv'_{k+i}$. One can check that  

\begin{align*}
\textbf{u}(x)\circ\textbf{v}(x)&=\sum^{\alpha-1}_{n=0}\left[ 2\theta_\frac{\mathfrak{m}}{\alpha}(x^\alpha)\sum^{\alpha-1}_{j=0}u_jv_{j+n}x^{\mathfrak{m}-1-n}\right] +\\ 
& \sum^{\beta-1}_{t=0}\left[  \theta_\frac{\mathfrak{m}}{\beta}(x^\beta)\sum^{\beta-1}_{k=0}u'_kv'_{k+t}x^{\mathfrak{m}-1-t}\right]\mod(x^{\mathfrak{m}}-1)\\
&=\theta_\frac{\mathfrak{m}}{\alpha}(x^\alpha)\left[\sum^{\alpha-1}_{n=0} 2\sum^{\alpha-1}_{j=0}u_jv_{j+n}x^{\mathfrak{m}-1-n}\right] +\\ 
& \theta_\frac{\mathfrak{m}}{\beta}(x^\beta)\left[\sum^{\beta-1}_{t=0}  \sum^{\beta-1}_{k=0}u'_kv'_{k+t}x^{\mathfrak{m}-1-t}\right]\mod(x^{\mathfrak{m}}-1).
\end{align*}

Then, arranging the terms one obtains that

$$\textbf{u}(x)\circ\textbf{v}(x)=\sum^{\mathfrak{m}-1}_{i=0} S_i x^{\mathfrak{m}-1-i}\mod(x^{\mathfrak{m}}-1).$$
Thus, $\textbf{u}(x)\circ\textbf{v}(x)=0$ if and only if $S_i=0$ for $0\leq i\leq \mathfrak{m}-1.$
\end{proof}


\begin{lemma}\label{Lemma1}
Let $\textbf{u}=(u(x)\mid u'(x))$ and $\textbf{v}(x)=(v(x)\mid v'(x))$ be elements in $R_{\alpha,\beta}$ such that {$\textbf{u}(x)\circ\textbf{v}(x)=0$}. If $u'(x)$ or $v'(x)$ equals $0$, then $u(x)v^*(x)\equiv 0\pmod{(x^\alpha-1)}$ over $\mathbb{Z}_2$. If $u(x)$ or $v(x)$ equal $0$, then $u'(x)v'^*(x)\equiv 0\pmod{(x^\beta-1)}$ over $\mathbb{Z}_4$.
\end{lemma}

\begin{proof}
Let $u'(x)$ or $v'(x)$ equal $0$, then 
$$0=\textbf{u}(x)\circ\textbf{v}(x)=2u(x)\theta_{\frac{\mathfrak{m}}{\alpha}}(x^\alpha)x^{\mathfrak{m}-1-\deg(v(x))}v^*(x)+0 \mod(x^{\mathfrak{m}}-1).$$ 
So, 
$$2u(x)\theta_{\frac{\mathfrak{m}}{\alpha}}(x^\alpha)x^{\mathfrak{m}-1-\deg(v(x))}v^*(x)=2\mu'(x)(x^\mathfrak{m}-1),$$
for some $\mu'(x)\in\mathbb{Z}_4[x]$. 

This is equivalent to 
$$u(x)\theta_{\frac{\mathfrak{m}}{\alpha}}(x^\alpha)x^{\mathfrak{m}-1-\deg(v(x))}v^*(x)=\mu'(x)(x^\mathfrak{m}-1)\in\mathbb{Z}_2[x].$$

By Proposition \ref{xnm=xntheta},
$$u(x)x^{\mathfrak{m}}v^*(x)=\mu(x)(x^\alpha-1),$$
$$u(x)v^*(x)\equiv 0\pmod{(x^\alpha-1)}.$$
A similar argument can be used to prove the other case.
\end{proof}


The following proposition determines the degrees of the generator 
polynomials of the dual code in terms of the degrees of the generator polynomials of 
the code. These results will be helpful to determine the 
generator polynomials of the dual code.

\begin{proposition}\label{DualPolynomialDegrees}\label{deg_bar_b}
Let ${\cal C}=\langle (b\mid { 0}), (\ell \mid  fh +2f) \rangle$ be a ${\mathbb{Z}_2 {\mathbb{Z}_4}}$-additive cyclic code of type $(\alpha, \beta; \gamma , \delta; \kappa)$, where $fgh=x^\beta-1$, and with dual code ${\cal C}^\perp = \langle (\bar{b}\mid { 0}), (\bar{\ell} \mid  \bar{f}\bar{h} +2\bar{f}) \rangle ,$ where $\bar{f}\bar{g}\bar{h} = x^\beta -1.$
Then, 
\begin{align*}
\deg(\bar{b})\! &=\!  \alpha - \deg(\gcd(b,\ell)),\\
\deg(\bar{f})\! &=\! \deg(g)\!+\!\deg(\gcd(b,\ell))\!-\!\deg(\gcd(b,\ell g)),\\
\deg(\bar{h})\! &=\! \deg(h)\!-\!\deg(b)\!-\!\deg(\gcd(b,\ell))\!+\!2\deg(\gcd(b,\ell g)),\\
\deg(\bar{g})\! &=\! \deg(f)\!+\!\deg(b)\!-\!\deg(\gcd(b,\ell g)).\\
\end{align*} 
\end{proposition}
\begin{proof}
Let ${\cal C}^\perp$ be a code of type $(\alpha, \beta; \bar{\gamma }, 
\bar{\delta}; \bar{\kappa})$. It is easy to prove that $({\cal C}_X)^\perp$ is a 
binary cyclic code generated by $\bar{b}$, so $|({\cal 
C}_X)^\perp|=2^{\alpha-\deg(\bar{b})}$. Moreover, by Proposition 
\ref{Dimension_subcodes}, $|({\cal C}_X)^\perp|=2^{\alpha-\kappa-\delta_1}$ and 
by Proposition \ref{k1k2d1d2}, we obtain that $\deg(\bar{b}) =  
\alpha-\deg(\gcd(b,\ell)).$ Finally, from \cite{AddDual} it is known that
\begin{align*}
\bar{\gamma} &= \alpha +\gamma -2\kappa,\\
\bar{\delta} &= \beta -\gamma -\delta +\kappa,\\
\bar{\kappa} &= \alpha - \kappa,
\end{align*}
and applying Theorem \ref{TypeDependingDeg} to the parameters of ${\cal C}$ and ${\cal C}^\perp$, we obtain the result.
\end{proof}

We know that a $\mathbb{Z}_2\mathbb{Z}_4$-additive code ${\cal C}$ is separable if and only if ${\cal C}^\perp$ is separable. 
Moreover, if a $\mathbb{Z}_2\mathbb{Z}_4$-additive cyclic code is separable, then it is easy to find the generator polynomials of the dual, that are given in the following proposition.

\begin{proposition}
Let ${\cal C}=\langle (b\mid { 0}), ({ 0} \mid fh +2f) \rangle$ be a separable ${\mathbb{Z}_2 {\mathbb{Z}_4}}$-additive cyclic code of type $(\alpha, \beta; \gamma , \delta; \kappa)$, where $fgh=x^\beta-1$. Then, $${\cal C}^\perp = \langle (\frac{x^\alpha-1}{b^*}\mid { 0}), ({ 0} \mid  g^*h^* +2g^*) \rangle.$$
\end{proposition}
\begin{proof}
If ${\cal C}$ is separable, then ${\cal C}^ \perp = ({\cal C}_X)^\perp \times ({\cal C}_Y)^\perp$, where $({\cal C}_X)^\perp= \langle \frac{x^\alpha-1}{b^*}\rangle$ and $({\cal C}_Y)^\perp = \langle g^*h^* +2g^* \rangle$.
\end{proof}


\begin{proposition}\label{bar(b)}
Let ${\cal C}=\langle (b\mid { 0}), (\ell \mid fh +2f) \rangle$ be a ${\mathbb{Z}_2 {\mathbb{Z}_4}}$-additive cyclic code of type $(\alpha, \beta; \gamma , \delta; \kappa)$ with dual code ${\cal C}^\perp = \langle (\bar{b}\mid { 0}), (\bar{\ell} \mid \bar{f}\bar{h} +2\bar{f}) \rangle.$
Then, $$\bar{b} = \frac{x^\alpha -1}{(gcd(b,\ell))^*}\in\mathbb{Z}_2[x].$$
\end{proposition}

\begin{proof}
We have that $(\bar{b}\mid { 0})$ belongs to ${\cal C}^\perp$. Then,
\begin{align*}
(b\mid { 0})\circ (\bar{b}\mid { 0}) = & 0 
,\\
(\ell \mid fh +2f)\circ (\bar{b}\mid { 0}) = & 0 
.
\end{align*}
Therefore, by Lemma \ref{Lemma1},
\begin{align*}
b\bar{b}^*\equiv & 0 \pmod{(x^\alpha-1)}
,\\
\ell \bar{b}^*\equiv & 0 \pmod{(x^\alpha-1)}
,
\end{align*}
over $\mathbb{Z}_2$. So, $ \gcd(b,\ell)\bar{b}^*\equiv 0 \pmod{(x^\alpha-1)}$, and there exist $\mu\in \mathbb{Z}_2[x]$ such that $ \gcd(b,\ell)\bar{b}^*= \mu(x^\alpha-1)$. \\
Moreover, since $\gcd(b,\ell)$ and $\bar{b}^*$ divides $(x^\alpha-1)$ and, by Proposition \ref{deg_bar_b}, we have that $\deg(\bar{b}) =  \alpha-\deg(\gcd(b,\ell)).$ We conclude that
$$\bar{b}^*=\frac{x^\alpha-1}{\gcd(b,\ell)}\in\mathbb{Z}_2[x].$$
\end{proof}

\begin{proposition}\label{bar_fh}
Let ${\cal C}=\langle (b\mid { 0}), (\ell \mid fh +2f) \rangle$ be a ${\mathbb{Z}_2 {\mathbb{Z}_4}}$-additive cyclic code of type $(\alpha, \beta; \gamma , \delta; \kappa)$, where $fgh=x^\beta-1$, and with dual code ${\cal C}^\perp = \langle (\bar{b}\mid { 0}), (\bar{\ell} \mid \bar{f}\bar{h} +2\bar{f}) \rangle ,$ where $\bar{f}\bar{g}\bar{h} = x^\beta -1.$
Then, 
$\bar{f}\bar{h}$ is the Hensel lift of the polynomial $\frac{(x^\beta-1)\gcd(b,\ell g)^*}{f^*b^*}\in\mathbb{Z}_2[x].$
\end{proposition}
\begin{proof}
It is known that $h$ and $g$ are coprime, from which we deduce easily that $p_1fh+p_2fg=f$, for some $p_1,p_2\in\mathbb{Z}_4[x]$. Since $(b\mid 0)$, $(0\mid 2fh)$ and $(\ell g\mid 2fg)$ belong to ${\cal C}$, then 
$$(0\mid \frac{b}{\gcd(b,\ell g)}(2p_1fh+2p_2fg))=(0\mid \frac{b}{\gcd(b,\ell g)}2f)\in{\cal C}.$$

Therefore,
$$(\bar{\ell} \mid \bar{f}\bar{h} +2\bar{f})\circ (0\mid \frac{b}{\gcd(b,\ell g)}2f)=0 .$$
Thus, by Lemma \ref{Lemma1},
$$(\bar{f}\bar{h} +2\bar{f})\left( \frac{b^*2f^*}{\gcd(b,\ell g)^*}\right)\equiv 0\pmod{(x^\beta -1)},$$
and
\begin{equation}\label{equ1}
(2\bar{f}\bar{h})\left( \frac{b^*f^*}{\gcd(b,\ell g)^*}\right)=2\mu (x^\beta -1),
\end{equation}
for some $\mu\in\mathbb{Z}_4[x].$

If (\ref{equ1}) holds over $\mathbb{Z}_4$, then it is equivalent to
$$({\bar{f}\bar{h}})\left( \frac{b^*{f^*}}{\gcd(b,\ell g)^*}\right)={\mu} (x^\beta -1)\in\mathbb{Z}_2[x]$$

It is known that ${\bar{f}\bar{h}}$ is a divisor of $x^\beta-1$ and, by Corollary \ref{bdiviXSgcd}, we have that $\left( \frac{b^*{f^*}}{\gcd(b,\ell g)^*}\right)$ divides $(x^\beta-1)$ over $\mathbb{Z}_2$. By Corollary \ref{DualPolynomialDegrees}, $\deg({\bar{f}\bar{h}})=\beta-\deg(f)-\deg(b)+\deg(\gcd(b,\ell g))$, so 
$$\beta=\deg\left({\bar{f}\bar{h}} \frac{b^*{f^*}}{\gcd(b,\ell g)^*}\right)=\deg(x^\beta-1).$$
Hence, we obtain that ${\mu}=1\in\mathbb{Z}_2$ and
\begin{equation}\label{HenLift_fh}
{\bar{f}\bar{h}}=\frac{(x^\beta-1)\gcd(b,\ell g)^*}{{f^*}{b^*}}\in\mathbb{Z}_2[x].
\end{equation}

Since $\beta$ is odd and by the uniqueness of the Hensel lift \cite[p.73]{Wan}, $\bar{f}\bar{h}$ is the unique monic polynomial in $\mathbb{Z}_4[x]$ dividing $(x^\beta-1)$ and satisfying (\ref{HenLift_fh}).
\end{proof}


\begin{proposition}\label{bar_f}
Let ${\cal C}=\langle (b\mid { 0}), (\ell \mid fh +2f) \rangle$ be a ${\mathbb{Z}_2 {\mathbb{Z}_4}}$-additive cyclic code of type $(\alpha, \beta; \gamma , \delta; \kappa)$, where $fgh=x^\beta-1$, and with dual code ${\cal C}^\perp = \langle (\bar{b}\mid { 0}), (\bar{\ell} \mid \bar{f}\bar{h} +2\bar{f}) \rangle ,$ where $\bar{f}\bar{g}\bar{h} = x^\beta -1.$
Then, 
$\bar{f}$ is the Hensel lift of the polynomial $\frac{(x^\beta-1)\gcd(b,\ell)^*}{f^*h^*\gcd(b,\ell g)^*}\in\mathbb{Z}_2[x].$
\end{proposition}
\begin{proof}
One can factorize in $\mathbb{Z}_2[x]$ the polynomials $b, \ell, \ell g$ in 
the  following way:
\begin{align*}
\ell =&\gcd(b, \ell) \rho,\\
\ell g =&\gcd(b,\ell g)\rho\tau_1,\\
b =& \gcd(b, \ell g) \tau_2,
\end{align*}
where $\tau_1$ and $\tau_2$ are coprime polynomials. 

Hence, there exist 
$t_1,t_2\in\mathbb{Z}_2[x]$ such that $t_1\tau_1+t_2\tau_2=1.$ Then,
$$\gcd(b,\ell g)\rho(t_1\tau_1 + t_2 \tau_2)=\gcd(b,\ell g)\rho,$$
and
$$ t_1\ell g + \rho t_2 b=\frac{\gcd(b,\ell g)}{\gcd(b, \ell)}\ell.$$

Therefore,
{ 
\begin{align*}
\frac{\gcd(b,\ell g)}{\gcd(b, \ell)}\star&(\ell\mid fh+2f)+t_1\star(\ell g\mid 2fg)+\rho t_2\star(b\mid 0)=\\
& \left(0\mid \frac{\gcd(b,\ell g)}{\gcd(b, \ell)}(fh+2f)+t_1 2fg\right)\in{\cal C}.
\end{align*}
}

Since $\bar{h}$ and $\bar{g}$ are coprime, there exist $\bar{p}_1,\bar{p}_2\in\mathbb{Z}_4[x]$ such that $2\bar{p}_1\bar{f}\bar{h}+2\bar{p}_2\bar{f}\bar{g}=2\bar{f}.$ So, $(2\bar{p}_1+\bar{p}_2\bar{g})\star(\bar{\ell}\mid \bar{f}\bar{h} + 2\bar{f})=(\bar{p}_2\bar{\ell}\bar{g}\mid 2\bar{f})\in{\cal C}^\perp.$

Therefore,
{
$$(\bar{p}_2\bar{\ell}\bar{g}\mid 2\bar{f})\circ\left(0\mid \frac{\gcd(b,\ell g)}{\gcd(b, \ell)}(fh+2f)+t_1 2fg\right)=0.$$
}

By Lemma \ref{Lemma1} and, arranging properly, we obtain that
$$2\bar{f}\left(\frac{\gcd(b,\ell g)^*}{\gcd(b, \ell)^*}\right) f^*h^*\equiv 0 \pmod{(x^\beta-1)}$$ 
and
\begin{equation}\label{equ2}
2\bar{f}\left(\frac{\gcd(b,\ell g)^*}{\gcd(b, \ell)^*}\right) f^*h^*= 2\mu(x^\beta-1),
\end{equation}
for some $\mu\in\mathbb{Z}_4[x].$ 

If (\ref{equ2}) holds over $\mathbb{Z}_4$, then it is equivalent to
$${\bar{f}}\left(\frac{\gcd(b,\ell g)^*}{\gcd(b, \ell)^*}\right) {f^*}{h^*}= {\mu}(x^\beta-1)\in\mathbb{Z}_2[x].$$

It is easy to prove that $\left(\frac{\gcd(b,\ell g)^*}{\gcd(b, \ell)^*}\right) {f^*}{h^*}$ divides $(x^\beta-1)$ in $\mathbb{Z}_2[x]$. By Corollary \ref{DualPolynomialDegrees}, $\deg(\bar{f})=\beta-\deg(f)-\deg(h)+\deg(\gcd(b,\ell))-\deg(\gcd(b,\ell g))$, so 
$$\beta=\deg\left({\bar{f}}\left(\frac{\gcd(b,\ell g)^*}{\gcd(b, \ell)^*}\right) {f^*}{h^*}\right)=\deg(x^\beta-1).$$
Hence, we obtain that ${\mu}=1$ and
\begin{equation}\label{HenLift_f}
{\bar{f}}= \frac{(x^\beta-1)\gcd(b, \ell)^*}{\gcd(b,\ell g)^*{f^*}{h^*}}\in\mathbb{Z}_2[x].
\end{equation}
Since $\beta$ is odd and by the uniqueness of the Hensel lift \cite[p.73]{Wan} then $\bar{f}$ is the unique monic polynomial in $\mathbb{Z}_4[x]$ dividing $(x^\beta-1)$ and holding (\ref{HenLift_f}).
\end{proof}

\begin{lemma}\label{B/GCDdiviH}
Let ${\cal C}=\langle (b\mid { 0}), (\ell \mid fh +2f) \rangle$ be a ${\mathbb{Z}_2 {\mathbb{Z}_4}}$-additive cyclic code of type $(\alpha, \beta; \gamma , \delta; \kappa)$, where $fgh=x^\beta-1$. Then, the Hensel lift of $\frac{b}{\gcd(b,\ell g)}$ divides $h$.
\end{lemma}
\begin{proof}
In general, if $a\mid b\mid x^\beta-1$ over $\mathbb{Z}_2[x]$ with $\beta$ odd, then the Hensel lift of $a$ divides the Hensel lift of $b$ that divides $x^\beta-1$ over $\mathbb{Z}_4[x]$. Then, by Corollary \ref{bdiviXSgcd}, the result follows.
\end{proof}

In the family of $\mathbb{Z}_2\mathbb{Z}_4$-additive cyclic codes there is a particular class when the polynomials $b$ and $\gcd(b, \ell g)$ are the same. Applying \mbox{Lemma \ref{generators_Cb}} to this class we obtain that ${\cal C}_b$ has only two generators, $\langle (b\mid 0), (0\mid 2f)\rangle$, instead of three, $\langle (b\mid { 0}), (\ell g \mid  2fg), ({ 0}\mid  2fh) \rangle$. So, we have to take care of this class of $\mathbb{Z}_2\mathbb{Z}_4$-additive cyclic codes.

\begin{proposition}
Let ${\cal C}=\langle (b\mid { 0}), (\ell \mid fh +2f) \rangle$ be a 
non-separable ${\mathbb{Z}_2 {\mathbb{Z}_4}}$-additive cyclic code of type 
$(\alpha, \beta; \gamma , \delta; \kappa)$, where $fgh=x^\beta-1$, and with dual 
code ${\cal C}^\perp = \langle (\bar{b}\mid{ 0}), (\bar{\ell} \mid 
\bar{f}\bar{h} +2\bar{f}) \rangle ,$ where $\bar{f}\bar{g}\bar{h} = x^\beta -1.$ 
Let $\rho=\frac{\ell }{\gcd(b,\ell)}$. 
Then, 
$$\bar{\ell}=\frac{x^\alpha-1}{b^*}\left(\frac{\gcd(b,\ell g)^*}{\gcd(b, \ell)^*} x^{\mathfrak{m} -\deg(f)}\mu_1+\frac{b^*}{\gcd(b,\ell g)^*} x^{\mathfrak{m} 
-\deg(fh)}\mu_2\right),$$ where
$$\left\{
 \begin{tabular}{l l l}
   $\mu_1 =x^{\deg(\ell)}(\rho^*)^{-1}\mod\left(\frac{b^*}{\gcd(b,\ell g)^*}\right),$ \\
   $\mu_2 = x^{\deg(\ell)}(\rho^*)^{-1}\mod\left(\frac{b^*}{\gcd(b,\ell)^*}\right).$ \\
\end{tabular}
  \right.$$
 
\end{proposition}
\begin{proof}
In order to calculate $\bar{\ell}$, by using $\circ$, we are going to operate $(\bar{\ell} \mid \bar{f}\bar{h} +2\bar{f})$ by three different codewords of $\cal{C}$. The result of these operations is $0$ modulo $x^\mathfrak{m}-1$.

First, consider $(\bar{\ell} \mid \bar{f}\bar{h} +2\bar{f})\circ (b\mid { 0}) = 0$
. By Lemma 
\ref{Lemma1}, $\bar{\ell}b^*\equiv 0\pmod{(x^\alpha-1)}$ and, for some 
$\lambda\in\mathbb{Z}_2[x],$ we have 
that $\bar{\ell}=\frac{x^\alpha-1}{b^*}\lambda.$

Second, consider $\tau=\frac{\gcd(b,\ell g)}{\gcd(b,\ell)}$ and compute $(\bar{\ell} 
\mid \bar{f}\bar{h} +2\bar{f})\circ (\tau\ell \mid \tau fh + 2\tau f).$ Let $t=\deg(\tau)$ and note that $(fh+2f)^*= f^*h^*+2x^{\deg(h)}f^*$. We obtain that 
\begin{align}
 0=(\bar{\ell} \mid \bar{f}&\bar{h} +2\bar{f})\circ \left(\tau\ell \mid \tau fh + 2\tau f\right)=\nonumber \\
 &2\bar{\ell} \theta_{\frac{\mathfrak{m}}{\alpha}}(x^\alpha)x^{\mathfrak{m}-\deg(\ell)-1-t}\tau^*\ell^*\nonumber \\
 &+  \bar{f}\bar{h}\theta_{\frac{\mathfrak{m}}{\beta}}(x^\beta) x^{\mathfrak{m}-\deg(fh)-1-t}\tau^* f^*h^*\label{bar(fh)fh0}\\
 &+ 2\bar{f}\bar{h}\theta_{\frac{\mathfrak{m}}{\beta}}(x^\beta)x^{\mathfrak{m}-\deg(f)-1-t}\tau^* f^*\nonumber \\
 &+ 2\bar{f} \theta_{\frac{\mathfrak{m}}{\beta}}(x^\beta)x^{\mathfrak{m}-\deg(fh)-1-t}\tau^* f^*h^* \mod(x^{\mathfrak{m}}-1).\nonumber 
\end{align}

Apply Proposition \ref{xnm=xntheta} to each addend and 
$\bar{\ell}=\frac{x^\alpha-1}{b^*}\lambda.$ In addend (\ref{bar(fh)fh0}), by 
Proposition \ref{bar_fh}, we may replace $\bar{f}\bar{h}$ by the Hensel lift of 
$\frac{(x^\mathfrak{\beta}-1)\gcd(b,\ell g)^*}{f^*b^*}$. The Hensel lift of 
$(x^\beta-1)$ and $f^*\pmod 2$ are the same polynomials $(x^\beta-1)$ and $f^*$. 
Moreover, by Lemma \ref{B/GCDdiviH}, the addend (\ref{bar(fh)fh0}) is 0 modulo 
$(x^\mathfrak{m}-1).$ Therefore, by Proposition \ref{bar_fh} and Proposition 
\ref{bar_f}, we get that
\begin{align}
 0=(\bar{\ell} \mid \bar{f}&\bar{h} +2\bar{f})\circ \left(\tau\ell \mid \tau fh + 2\tau f\right)=\nonumber \\
 &2\frac{(x^\mathfrak{m}-1)}{b^*}\lambda x^{\mathfrak{m}-\deg(\ell)-1-t}\tau^* \ell^*\nonumber \\ 
 &+ 2\frac{(x^\mathfrak{m}-1)\gcd(b,\ell g)^*}{f^*b^*}x^{\mathfrak{m}-\deg(f)-1-t}\tau^* f^*\nonumber \\
 &+ 2\frac{(x^\mathfrak{m}-1)\gcd(b,\ell)^*}{f^*h^*\gcd(b,\ell g)^*}x^{\mathfrak{m}-\deg(fh)-1-t}\tau^* f^*h^*\mod(x^{\mathfrak{m}}-1)\label{bar(f)fh0}.
\end{align}

Clearly, the addend (\ref{bar(f)fh0}) is 0 modulo $(x^\mathfrak{m}-1)$. Since $\tau=\frac{\gcd(b,\ell g)}{\gcd (b, \ell)}$, we have that $(\bar{\ell} \mid \bar{f}\bar{h} +2\bar{f})\circ \left(\tau\ell \mid \tau fh + 2\tau f\right)$ is equal to
\begin{equation}\label{2smth}
2\frac{(x^\mathfrak{m}-1)\gcd(b,\ell g)^*}{b^*}\left(\lambda x^{\mathfrak{m}-\deg(\ell)-1-t}\rho^* + x^{\mathfrak{m}-\deg(f)-1-t}\tau^* \right)\equiv 0\pmod{(x^\mathfrak{m}-1)}.
\end{equation}
This is equivalent, over $\mathbb{Z}_2$, to\\
{
\begin{equation*}
\frac{(x^\mathfrak{m}-1)\gcd(b,\ell g)^*}{b^*}\left(\lambda x^{\mathfrak{m}-\deg(\ell)-1-t}\rho^* + x^{\mathfrak{m}-\deg(f)-1-t}\tau^* \right)\equiv 0\pmod{(x^\mathfrak{m}-1)}. 
\end{equation*}
}
Then,
\begin{equation}\label{lambdamodxm_neq_1}
\left(\lambda x^{\mathfrak{m}-\deg(\ell)-1-t}\rho^* + x^{\mathfrak{m}-\deg(f)-1-t}\tau^* \right)\equiv 0\pmod{(x^\mathfrak{m}-1)},
\end{equation}
or
{
\begin{equation}\label{lambdamodb_neq_1}
\left(\lambda x^{\mathfrak{m}-\deg(\ell)-1-t}\rho^* + x^{\mathfrak{m}-\deg(f)-1-t}\tau^* \right)\equiv 0 \!\!\pmod{\!\left(\frac{b^*}{\gcd(b,\ell g)^*}\right)}.
\end{equation}
}
Since $(\frac{b^*}{\gcd(b,\ell g)^*})$ divides $(x^ \mathfrak{m}-1)$, then (\ref{lambdamodxm_neq_1}) implies (\ref{lambdamodb_neq_1}). 

The greatest common divisor between $\rho$ and $\left(\frac{b}{\gcd(b,\ell g)}\right)$ is $1$, then $\rho^*$ is invertible modulo $\left(\frac{b^*}{\gcd(b,\ell g)^*}\right).$ Thus,
$$\lambda= \tau^*  x^{\mathfrak{m} -\deg(f)+\deg(\ell 
)}(\rho^*)^{-1}\mod\left(\frac{b^*}{\gcd(b,\ell g)^*}\right).$$

Let  $\lambda_1=\tau^* x^{\mathfrak{m} -\deg(f)+\deg(\ell 
)}(\rho^*)^{-1}\mod\left(\frac{b^*}{\gcd(b,\ell g)^*}\right)$. Then $\lambda=\lambda_1+\lambda_2$ with $\lambda_2\equiv 0 \pmod{\left(\frac{b^*}{\gcd(b,\ell g)^*}\right)}.$

Finally, we compute $(\bar{\ell} \mid \bar{f}\bar{h} +2\bar{f})\circ (\ell \mid fh+2f).$ 
\begin{align}
 0=(\bar{\ell} \mid \bar{f}&\bar{h} +2\bar{f})\circ \left(\ell \mid fh + 2 f\right)=\nonumber \\
 &2\bar{\ell} \theta_{\frac{\mathfrak{m}}{\alpha}}(x^\alpha)x^{\mathfrak{m}-\deg(\ell)-1}\ell^*\nonumber \\
 &+  \bar{f}\bar{h}\theta_{\frac{\mathfrak{m}}{\beta}}(x^\beta) x^{\mathfrak{m}-\deg(fh)-1} f^*h^*\label{bar(fh)fh1}\\
 &+ 2\bar{f}\bar{h}\theta_{\frac{\mathfrak{m}}{\beta}}(x^\beta)x^{\mathfrak{m}-\deg(f)-1} f^*\nonumber \\
 &+ 2\bar{f} \theta_{\frac{\mathfrak{m}}{\beta}}(x^\beta)x^{\mathfrak{m}-\deg(fh)-1} f^*h^*\mod(x^{\mathfrak{m}}-1).\nonumber 
\end{align}
Apply Proposition \ref{xnm=xntheta} to each addend. By Lemma \ref{B/GCDdiviH} and replacing $\bar{f}\bar{h}$ by the Hensel lift of $\frac{(x^\mathfrak{\beta}-1)\gcd(b,\ell g)^*}{f^*b^*}$, then the addend (\ref{bar(fh)fh1}) is $0 \mod(x^\mathfrak{m} -1)$ and, by Proposition \ref{bar_fh} and Proposition \ref{bar_f}, $ (\bar{\ell} \mid \bar{f}\bar{h} +2\bar{f})\circ \left(\ell \mid fh + 2 f\right)$ is equal to 

\begin{align*}
2\frac{(x^\mathfrak{m}-1)}{b^*}(\lambda_1+\lambda_2) x^{\mathfrak{m}-\deg(\ell)-1}\ell^*
 &+ 2\frac{(x^\mathfrak{m}-1)\gcd(b,\ell g)^*}{b^*} x^{\mathfrak{m}-\deg(f)-1}\\
 &+ 2\frac{(x^\mathfrak{m}-1)\gcd(b,\ell)^*}{\gcd(b,\ell g)^*}x^{\mathfrak{m}-\deg(fh)-1}
 \equiv 0\pmod{(x^\mathfrak{m}-1)}.
\end{align*}

Since $\lambda_1 = \tau^* x^{\mathfrak{m} -\deg(f)+\deg(\ell)}(\rho^*)^{-1}\mod\left(\frac{b^*}{\gcd(b,\ell g)^*}\right)$, we have that $$2\frac{(x^\mathfrak{m}-1)}{b^*}\lambda_1 x^{\mathfrak{m}-\deg(\ell)-1}\ell^*+ 2\frac{(x^\mathfrak{m}-1)\gcd(b,\ell g)^*}{b^*}x^{\mathfrak{m}-\deg(f)-1}\equiv 0\pmod{(x^{\mathfrak{m}}-1)}.$$
Therefore, we obtain that 

$$ 2\frac{(x^\mathfrak{m}-1)}{b^*}\lambda_2 x^{\mathfrak{m}-\deg(\ell)-1}\ell^*
 + 2\frac{(x^\mathfrak{m}-1)\gcd(b,\ell)^*}{\gcd(b,\ell g)^*}x^{\mathfrak{m}-\deg(fh)-1}\equiv 0\pmod{(x^\mathfrak{m}-1)},$$ 
and then
$$ 2\frac{(x^\mathfrak{m}-1)\gcd(b, \ell)^*}{b^*}\left(\lambda_2 x^{\mathfrak{m}-\deg(\ell)-1}\rho^*
 + \frac{b^*}{\gcd(b,\ell g)^*}x^{\mathfrak{m}-\deg(fh)-1}\right)\equiv 0\pmod{(x^\mathfrak{m}-1)}.$$ 
Arguing similar to the calculation of $\lambda$ in (\ref{2smth}), we obtain that
$$ \lambda_2=\frac{b^*}{\gcd(b,\ell g)^*} x^{\mathfrak{m} -\deg(fh)+\deg(\ell)}(\rho^*)^{-1}\mod\left(\frac{b^*}{\gcd(b,\ell)^*}\right).$$
Now, considering the values of $\lambda_1$ and $\lambda_2$ and defining properly $\mu_1$ and $\mu_2$ we obtain the expected result.
\end{proof}

We summarize the previous results in the next theorem.
\begin{theorem}\label{DualPolynomials}
Let ${\cal C}=\langle (b\mid { 0}), (\ell \mid fh +2f) \rangle$ be a 
\mbox{${\mathbb{Z}_2 {\mathbb{Z}_4}}$-additive} cyclic code of type $(\alpha, 
\beta; \gamma , \delta; \kappa)$ , where $fgh=x^\beta-1$, and with dual code 
${\cal C}^\perp = \langle (\bar{b}\mid { 0}), (\bar{\ell} \mid \bar{f}\bar{h} 
+2\bar{f}) \rangle ,$ where $\bar{f}\bar{g}\bar{h} = x^\beta -1.$ Let 
$\rho=\frac{\ell }{\gcd(b,\ell)}.$
Then, 
\begin{center}
\begin{enumerate}
\item $\bar{b} = \frac{x^\alpha -1}{(gcd(b,\ell))^*}\in\mathbb{Z}_2[x],$
\item $\bar{f}\bar{h}$ is the Hensel lift of the polynomial \mbox{$\frac{(x^\beta-1)\gcd(b,\ell g)^*}{f^*b^*}\in\mathbb{Z}_2[x].$}
\vspace*{1mm}
\item $\bar{f}$ is the Hensel lift of the polynomial \mbox{$\frac{(x^\beta-1)\gcd(b,\ell)^*}{f^*h^*\gcd(b,\ell g)^*}\in\mathbb{Z}_2[x].$}
\vspace*{1mm}
\item $\bar{\ell}=\frac{x^\alpha-1}{b^*}\left(\frac{\gcd(b,\ell g)^*}{\gcd(b, \ell)^*} x^{\mathfrak{m} -\deg(f)}\mu_1+\frac{b^*}{\gcd(b,\ell g)^*} x^{\mathfrak{m} 
-\deg(fh)}\mu_2\right)\in\mathbb{Z}_2[x],$ where
$$\left\{
 \begin{tabular}{l l l}
   $\mu_1=x^{\deg(\ell)}(\rho^*)^{-1}\mod\left(\frac{b^*}{\gcd(b,\ell g)^*}\right)$, \\
   $\mu_2=x^{\deg(\ell)}(\rho^*)^{-1}\mod\left(\frac{b^*}{\gcd(b,\ell)^*}\right).$ \\
\end{tabular}
  \right.$$
\end{enumerate}
\end{center}

\end{theorem}

Note that from Theorem \ref{DualPolynomials} and Theorem \ref{SpanningSetTh} one can easily compute the minimal spanning set of the dual code ${\cal C}^\perp$ as a $\mathbb{Z}_4$-module, and use the encoding method for ${\mathbb{Z}_2 {\mathbb{Z}_4}}$-additive cyclic codes described in \cite{Abu}.

\section{Examples}

As a simple example, consider the $\mathbb{Z}_2\mathbb{Z}_4$-additive cyclic code ${\cal C}_1 = 
\langle (x-1 \mid (x^2+x+1) + 2)\rangle$ of type $(3,3;2,1;2)$.  We have that 
$ b=x^3-1, \ell=(x-1),f=1$ and $h=x^2+x+1.$

The generator matrix (\cite{AddDual}) is
\[
G=\left(
\begin{array}{l|r}
 101&200\\
 011&220\\
 \hline
 000&111\\
\end{array}
 \right)
\]

Then, applying the formulas of Theorem \ref{DualPolynomials} we have
$\bar{b}=x^2+x+1, 
\bar{\ell}=x,\bar{f}\bar{h}=x-1,$ and $\bar{f}=x-1.$
Therefore,
${\cal C}_1^\perp=\langle(x^2+x+1\mid 0),(x\mid (x-1)+2(x-1))\rangle,$ is of 
type $(3,3;1,2;1)$ and has generator matrix
\[
H=\left(
\begin{array}{l|r}
 111&000\\
  \hline
 100&310\\
 001&301\\
\end{array}
 \right).
\]
 
\bigskip
In order to determine some cyclic codes with good parameters, we will consider 
some optimal codes with respect to the minimum distance. Applying the classical 
Singleton bound \cite{sing} to a $\mathbb{Z}_2\mathbb{Z}_4$-additive code ${\cal C}$ of type 
$(\alpha,\beta; \gamma,\delta;\kappa)$ and minimum distance $d$, the following 
bound is obtained:
\begin{equation}\label{MDSbound}
\frac{d-1}{2}\leq \frac{\alpha}{2}+\beta-\frac{\gamma}{2}-\delta.
\end{equation}
According to \cite{MDS}, a code meeting the bound (\ref{MDSbound}) is called maximum 
distance separable with respect to the Singleton bound, briefly MDSS.

By \cite[Theorem 19]{Abu} it is known that ${\cal C}=\langle (b\mid { 0}), 
(\ell \mid fh +2f) \rangle$ with $b=x-1$, $\ell=1$ and $f=h=1$ is an MDSS code 
of type $(\alpha,\beta; \alpha-1, \beta; \alpha-1)$. Applying Theorem 
\ref{DualPolynomials} to compute the dual code of ${\cal C}$ one obtain that 
${\cal C}^\perp = \langle (\bar{b}\mid { 0}), (\bar{\ell} \mid \bar{f}\bar{h} 
+2\bar{f}) \rangle$ with $\bar{b}=x^\alpha-1$, $\bar{\ell}=\theta_\alpha (x)$, 
$\bar{f}=\theta_\beta (x)$ and $\bar{h}=x-1$, which is also an MDSS code. In 
fact, the binary image of ${\cal C}$ is the set of all even weight vectors and 
the binary image of ${\cal C}^\perp$ is the repetition code. Moreover, these are 
the only MDSS $\mathbb{Z}_2\mathbb{Z}_4$-additive codes with more than one codeword and minimum 
distance $d>1$, as can be seen in \cite{MDS}.

\bigskip
Finally, we are going to see a pair of examples of self-dual 
$\mathbb{Z}_2\mathbb{Z}_4$-additive cyclic codes, giving the generators and type of these 
codes.

\begin{center}
\begin{tabular}{||p{10cm}|c||}
\hline\hline
Generators & Type  \\
\hline
$b=x^{10} + x^8 + x^7 + x^3 + x + 1, \ell=x^6 + x^4 + x + 1, fh=y^4 + 2y^3 + 
3y^2 + y + 1, f= 1$  & ( 14, 7; 8, 3; 7 ) \\ 
\hline
$b=x^5 + 1, \ell=0, fh=y^5 - 1, f= 1$  & ( 10, 5; 10, 0; 5 )  \\ 

\hline\hline
\end{tabular}
\end{center}

The second code in the table belongs to an infinite family of self-dual 
$\mathbb{Z}_2\mathbb{Z}_4$-additive cyclic codes that was given in \cite[Theorem 
4]{z2z4SD}.

\begin{proposition}
Let $\alpha$ be even and $\beta$ odd. Let ${\cal C}=\langle (b\mid { 0}), (\ell \mid fh +2f) \rangle$ be a \mbox{${\mathbb{Z}_2 {\mathbb{Z}_4}}$-additive} cyclic code with $b=x^{\frac{\alpha}{2}}-1$, $\ell= 0$, $h=x^\beta-1$ and $f= 1$. Then ${\cal C}$ is a self-dual code of type $(\alpha, \beta; \beta+\frac{\alpha}{2}, 0; \frac{\alpha}{2})$.
\end{proposition}
\begin{proof}
By Theorem \ref{DualPolynomials}, one obtains that 
$\bar{b}=x^{\frac{\alpha}{2}}-1$, $\bar{\ell}=0$, $\bar{h}=x^\beta-1$ and 
$\bar{f}=1.$ Hence ${\cal C}$ is self-dual and, by Theorem \ref{TypeDependingDeg}, 
it is of type $(\alpha, \beta; \beta+\frac{\alpha}{2}, 0; 
\frac{\alpha}{2})$.
\end{proof}

\end{document}